\documentclass[11pt]{amsart}
\usepackage[margin=1in]{geometry}
\usepackage{amsfonts}
\usepackage{amssymb}
\usepackage[dvips]{graphics}
\usepackage{epsfig}
\usepackage{euscript}
\usepackage{color}
\usepackage{cite}

 \newtheorem{theorem}{Theorem}[section]
 
 \newtheorem{lem}[theorem]{Lemma}
 \newtheorem{lemma}[theorem]{Lemma}
 
 \theoremstyle{definition}
 
 \theoremstyle{remark}

 \numberwithin{equation}{section}
 \newtheorem{proposition}[theorem]{Proposition}

\newcommand{\be}{\begin{equation}}
\newcommand{\ee}{\end{equation}}

\begin{document}
\title{Lifschitz Tails for Random Schr\"{o}dinger Operator in Bernoulli Distributed Potentials}
\date{\today}

\author{Michael Bishop}
\address{Department of Mathematics, University of California, Davis,
Davis, CA 95616}
\email{mbishop@math.ucdavis.edu}
\author{Vita Borovyk}
\address{Department of Mathematics, University of Cincinnati,
Cincinnati, OH 45221-0025}
\email{Vita.Borovyk@uc.edu}
\author{Jan Wehr}
\address{Department of Mathematics, University of Arizona,
Tucson, AZ 85721-0089}
\email{wehr@math.arizona.edu}

\begin{abstract}
This paper presents an elementary proof of Lifschitz tail behavior for random discrete Schr\"{o}dinger operators with a Bernoulli-distributed potential.  The proof approximates the low eigenvalues by eigenvalues of sine waves supported where the potential takes its lower value.  This is motivated by the idea that the eigenvectors associated to the low eigenvalues react to the jump in the values of the potential as if the gap were infinite.

\end{abstract}

\maketitle
\section{Introduction}
\label{sec:1}

In 1965, Lifschitz discovered that the density of states for disordered quantum systems exhibited an exponential structure at the ends of the spectrum \cite{Lifschitz65}; this property of random Schr\"{o}dinger operators is referred to as Lifschitz tail in his honor.  Random Schr\"{o}dinger operators take the form $-\Delta + V(x)$, where $\Delta$ is the Laplacian operator and $V(x)$ is a multiplication operator which is chosen from a random distribution; see \cite{Kirsch08} for a general introduction to random Schr\"{o}dinger operators.  Lifschitz argued that for an eigenvalue of an eigenfunction near the bottom of the spectrum, both the contribution from the Laplacian and the contribution from the multiplication operator $V$ must both be at the bottom of their respective spectra.  The former requires that the associated eigenfunction must be supported on a large set, while the latter requires that the value of the multiplication operator must be small for most of the support of the eigenfuntion.  The value of the potential function at a point in space is assumed to be random, independent of the value of the potential at other points, and bounded both above and below.  The probability that the potential function takes values near the bottom of its range for most sites in a specific large set is exponentially small in the size of the set.  In a large system, this probability forces the proportion of states near the bottom of the spectrum to also be exponentially small.  

The mathematical proofs that followed Lifschitz's discovery made precise this argument precise in a variety of settings \cite{BenderskiPastur,Pastur72,Fukushima74,FriedbergLuttinger,FukushimaNagaiNakao,Luttinger76,Nagai77,Nakao77,RomerioWreszinski,KirschMartinelli83,Simon85,KirschSimon86,Stollmann99,Schulz-Baldes,LuckNieu85,LuckNieu86,LuckNieu88}.  The proof presented in this paper follows a similar structure to \cite{Simon85}.  In \cite{Simon85}, the kinetic energy, the contribution to the eigenvalue due to the Laplacian, is bounded using Dirichlet-Neumann bracketing: the domain of the eigenfunctions is partitioned into boxes, the Laplacian operator is bounded by sums of operators defined on those boxes with Dirichlet or Neumann boundary conditions, then the operator bounds are used to bound the eigenvalues.  These bounds on the eigenvalues determine the size of a sufficiently large set to support an eigenstate with small eigenvalue.  Using large deviations theory, the probability that a given potential function will be sufficiently small on a sufficiently large box is shown to be exponentially small.  In a large system limit, the proportion of eigenfunctions near the bottom of the spectrum is also exponentially small.

The following proof of Lifschitz tail behavior for the one dimensional discrete random Schr\"{o}dinger operator with Bernoulli-distributed potential follows similar intuition to the proof structure described above with key steps replaced. The long intervals of zero potential act as support for the low energy states, replacing the boxes which partition the space in other proofs.  Instead of using Dirichlet-Neumann bracketing to bound the kinetic energy, the bounds are derived by optimizing the energy of sine waves on intervals of zero potential with boundary conditions weighted by the potential energy contribution of neighboring sites of positive potential.  The large deviations calculation used above is replaced by the discrete exponential distribution of the intervals (also known as a geometric distribution).  The picture that motivated this proof is that the excited state energies (the lowest eigenvalues) are approximated by the energies of sine waves supported on long intervals of zero potential.  This follows the way of thinking in \cite{BishopWehr12} where the ground state energy (the lowest eigenvalue) of this operator is approximated by the energy of a sine wave supported on the longest interval of sites where the potential is zero.  

This paper is structured as follows.  First, the operator, eigenvalues, and density of states are defined and the Lifschitz tail result is stated for the parameters defined.  Second, the operator is bounded above by an operator where states are excluded from sites of positive potential (essentially by making the potential infinite), bounding the excited state energies.  By estimating the distribution of intervals of zero potential, these upper bounds on energies are used to bound the density of states from below in the large system limit.  Third, the kinetic energy of a sine wave with weighted boundary conditions is bounded below, where this bound depends on the approximate integer frequency of the sine wave.  These lower bounds on energy provide an upper bound on the number of states a given interval can support with small energy.  Using the estimates on the distribution of intervals derived for the lower bound on the density of states, this upper bound on the number of states for a given interval is used to bound the density of states from above.

On the Hilbert space $\mathcal{H}=\ell^2\{0, \dots, L+1\}$, consider the operator $H = -\Delta + V$ with Dirichlet boundary conditions at $0$ and at $L+1$.  The Laplacian $-\Delta$ is defined as $-\Delta f(x) = 2f(x) - f(x-1) - f(x+1)$ and $V$ is a realization of $L$ i.i.d. (independent and identically distributed) random variables $V(x)$, $x = 1, \dots, L$.
on a probability space $(\Omega, \mathcal{F}, P)$ with the Bernoulli distribution
\begin{equation}
V(x) = \begin{cases}0 \quad \textrm{ with probability } p \\ b \quad \textrm{ with probability } q = 1 - p. \end{cases}
\end{equation} 

For each realization of the potential, the corresponding $H$ is self-adjoint, and therefore it has $L$ real eigenvalues counting the multiplicities:
\begin{equation}																			\label{eq:eigenvalues}
E_1 \le E_2 \le ... \le E_{L}.
\end{equation}
Each of the eigenvalues in the sequence is a random variable, however we suppress the dependence on $\omega \in \Omega$ in the notation.

Denote by $N_L(\epsilon)$ the eigenvalue counting function for the operator $H$: $N_L(\epsilon) = \#\{E_k: E_k < \epsilon\}$.  While $N_L(\epsilon)$ is a random variable, it self-averages in the infinite volume limit.  That is, with probability one
\begin{equation}
k(\epsilon) = \lim_{L \to \infty} \frac {N_L(\epsilon)}{L}
\end{equation}
exists and its value is independent of the realization of the (infinite-volume) potential. $k(\epsilon)$ is called the {\it integrated density of states}. See \cite{Kirsch08} for details. 

\begin{theorem}  For $\epsilon$ small enough the integrated density of states $k(\epsilon)$ satisfies
\begin{eqnarray}
\label{eq:LifTailslowerbnd} &&k(\epsilon) \geq\liminf_{L \to \infty} \frac{N_L(\epsilon)}{L} \geq \frac{qp^{\frac{\pi}{\sqrt{\epsilon}}}}{1 - p^{\frac{\pi}{\sqrt{\epsilon}}}}\\
\label{eq:LifTailsupperbnd}&&k(\epsilon) \leq \limsup_{L\to\infty} \frac{N_L(\epsilon)}{L} \leq \frac{q p^{\frac{\pi}{\sqrt{\epsilon}}- \frac {\pi^2}{b}}}{p^2\left(1-p^{\frac{\pi}{\sqrt{\epsilon}} + O(\sqrt{\epsilon})}\right)}
\end{eqnarray}
Equation (\ref{eq:LifTailsupperbnd}) holds for $\epsilon$ small enough. 
\end{theorem}

\section{Upper Bounds for Excited State Energies}
\label{sec:2}

The main result of this section is a set of simple upper bounds on $E_k$ which will be convenient in deriving the lower bound on the density of states (\ref{eq:LifTailslowerbnd}). We obtain them by comparing $H$ to a ``bigger" operator $\tilde{H} = -\Delta + \infty V$, where the potential takes values zero and infinity. The ``eigenstates" for $\tilde{H}$ are discrete sine waves (see (\ref{eq:discretesine})), supported on intervals of zero potential. Let $\{\ell_{\alpha, \beta}\}$ be a collection of the lengths of such intervals, where all $\ell_{\alpha, \beta}$ are equal to the same number $\ell_\alpha$ and the index $\beta$ counts the intervals of equal length. 
The energy of a sine wave supported on an interval of length $\ell_{\alpha, \beta}$ with a frequency $w$ is bounded by $\displaystyle \frac {w^2 \pi^2}{(\ell_{\alpha, \beta}+1)^2}$; these values form natural upper bounds for the eigenvalues of $H$. See \cite{BishopWehr12} for a derivation and a detailed discussion of the discrete sine waves. Note that a frequency of a sine wave supported on an interval of length $\ell_\alpha$ is an integer between $1$ and $\ell_{\alpha}$, thus each interval supports exactly $\ell_\alpha$ different sine waves.

We will use the following notation: for each $\alpha$ and $\beta$ introduce the set
\begin{equation}																														\label{eq:Ualphabeta}
\mathcal{U}_{\alpha, \beta}= \left\{\frac {w^2 \pi^2}{(\ell_{\alpha, \beta}+1)^2}: w = 1, \dots,  \ell_{\alpha, \beta}\right\}
\end{equation}
and denote by $\{U_k\}$ the sequence consisting of all elements of the sets $\mathcal{U}_{\alpha, \beta}$ arranged in the non-decreasing order, including repetitions. $\{U_k\}$ has $L'$ elements, where $L'$ denotes the number of sites where the given realization of the potential is zero.

The following theorem summarizes the above discussion.

\begin{theorem}																										\label{thm:Ekupperbounds}
Let (\ref{eq:eigenvalues}) be the eigenvalues of the operator $H$, corresponding to a specific realization of the potential $V$.
Then 
\begin{equation}																																				\label{eq:Ekupperbounds}
E_k \le U_k, \quad \textrm{ for all } \quad 1 \le k \le L'.
\end{equation} 
\end{theorem}
\begin{proof}: 
We first obtain an upper bound on $E_k$ through a minimax procedure (inequality (\ref{eq:Ekminimaxest})), and then show how \ref{eq:Ekupperbounds} follows from it. 

For the given realization of the potential, define a set $B = \{1 \le x \le L: V(x) = b\}$ and  let $\mathcal{H}_0$ be the subspace of  $\mathcal{H}$ defined by 
$\mathcal{H}_0 = \{\phi \in \mathcal{H}: \phi(x) = 0 \textrm{ for every } x \in B\}$. In other words, $\mathcal{H}_0$ consists of functions supported on the set of zeros of the potential. 

Next, recall that according to the minimax (or, in this form, maximin) principle (see for example \cite{Lax}),
\begin{equation}
E_k = \max_{T_{k-1}} \min_{\stackrel {\|\phi\| = 1,}{\phi \perp T_{k-1}}} \left\langle \phi, H \phi\right\rangle,
\end{equation} 
where the maximum is taken over all $(k-1)$-dimensional subspaces of $\mathcal{H}$. If $S_{k-1}$ denotes a $(k-1)$-dimensional subspace of $\mathcal{H}_0$, we have
\begin{equation}
\min_{\phi \perp T_{k-1}} \left\langle \phi, H \phi\right\rangle \le \min_{\phi \perp \mathrm{span}(S_{k-1}, \mathcal{H}_0^\perp) }\left\langle \phi, H \phi\right\rangle \quad \textrm{ if } \quad T_{k-1} \subset \mathrm{span}(S_{k-1}, \mathcal{H}_0^\perp).
\end{equation}
Since for every $T_{k-1} \subset \mathcal{H}$ there exists an $S_{k-1}$ such that $T_{k-1} \subset \mathrm{span}(S_{k-1}, \mathcal{H}_0^\perp)$, it follows that 
\begin{equation}
\min_{\phi \perp T_{k-1}} \left\langle \phi, H \phi\right\rangle \le \max_{S_{k-1} \subseteq \mathcal{H}_0} \min_{\phi \perp \mathrm{span}(S_{k-1}, \mathcal{H}_0^\perp) }\left\langle \phi, H \phi\right\rangle.
\end{equation}
Finally, taking $\max_{T_{k-1}}$ of both sides of the above equation and using  $$\{\phi \in \mathcal{H}: \phi \perp \mathrm{span}(S_{k-1}, \mathcal{H}_0^\perp)\} = \{\phi \in \mathcal{H}_0: \phi \perp S_{k-1}\},$$ we obtain the inequality
\begin{equation}																											\label{eq:Ekminimaxest}
E_k \le \max_{S_{k-1} \subseteq \mathcal{H}_0} \min_{\stackrel{\phi \in \mathcal{H}_0, \phi \perp S_{k-1},} {\|\phi\| = 1} }\left\langle \phi, H \phi\right\rangle, \quad  \textrm{ for each } \quad 1 \le k \le L'.
\end{equation}

Notice that the the right-hand side of the inequality (\ref{eq:Ekminimaxest}) is the expression for the $k$-th eigenvalue of the operator $H$ restricted to $\mathcal{H}_0$.  This operator acts as the negative Laplacian on the intervals where the potential is zero. Its eigenfunctions are sine waves, supported on individual intervals. More precisely, a discrete sine function with frequency $w$ supported on an interval $I = (x_0, x_0 +\ell+1)$, with zero boundary conditions at $x_0$ and at $x_0 +  \ell + 1$ is given by
\begin{equation}																					\label{eq:discretesine}
S(x) = \begin{cases} \sin \left(\frac {w \pi (x-x_0)}{\ell + 1}\right), &\quad \textrm{ if } x \in I \\ 0, &\quad \textrm{ if } x \notin I.\end{cases}
\end{equation}
The energy of $S$  is
\begin{equation}
4\sin^2\left( \frac {w \pi}{2 (\ell+1)}\right) \le \frac {w^2 \pi^2}{(\ell+1)^2}.
\end{equation}
Applying this estimate with $\ell = \ell_{\alpha, \beta}$ finishes the proof.
\end{proof}
	
\section{Lower Bound on the Lifschitz Tail}
\label{sec:3}
Here we will use the result of Section~\ref{sec:1} to obtain a lower bound on the density of states - inequality (\ref{eq:LifTailslowerbnd}). We start with an estimate on the counting function. 

\begin{lem} For $\epsilon$ small ($\epsilon \le \pi^2$ is enough), with probability one
\begin{equation}																							\label{eq:NLlowerbnd}
N_L(\epsilon) \ge \frac{q(1 - p^{\frac{\pi \ell_0}{\sqrt{\epsilon}}})p^{\frac{\pi}{\sqrt{\epsilon}}}}{p(1 - p^{\frac{\pi}{\sqrt{\epsilon}}})} \, L + o(L^\gamma), \quad \textrm{ as } L \to \infty,
\end{equation}
where $1/2 < \gamma < 1$ and $\ell_0 (= \ell_0(\omega))$ is the length of the longest interval of zero potential.
\end{lem}
\begin{proof}
We start with some pointwise estimates on $N_L(\epsilon)$, that follow from pointwise bounds on eigenvalues (Theorem~\ref{thm:Ekupperbounds}).
Keeping only the first $L'$ eigenvalues and dropping the rest will make the counting function smaller:
\begin{align*}
N_L(\epsilon) = \# \{E_k: E_k < \epsilon\} &\ge \#\{ E_k, 1\le k \le L': E_k < \epsilon\}\\
&\ge \#\{U_k: U_k < \epsilon\}.
\end{align*}
Here the second inequality is a direct consequence of the bounds (\ref{eq:Ekupperbounds}). It follows from the definition of $U_k$ (see equation (\ref{eq:Ualphabeta}) and the paragraph which contains it) that the value of the right-hand side of the above inequality is equal to the number of pairs $(\ell_{\alpha, \beta}, w)$ such that
\begin{equation}																									\label{eq:energyepsilonineq}
 \frac {w^2 \pi^2}{(\ell_{\alpha, \beta}+1)^2} < \epsilon.
\end{equation}
It is convenient to group such pairs according to frequencies $w$ rather than to the interval lengths, i.e. for each fixed $w$ we count the number of intervals $I$ with lengths satisfying (\ref{eq:energyepsilonineq}) (the length $|I|$ of an interval is one of the numbers $\ell_{\alpha, \beta}$). In other words, we ask how many intervals can support a state with frequency $w$ and energy less than $\epsilon$. This regrouping leads to the inequality
\begin{equation}																											\label{eq:NLineq1}
N_L(\epsilon) \ge \sum_{w = 1}^{\ell_0} \# \left\{I: \frac {w^2 \pi^2}{(|I|+1)^2} < \epsilon\right\} = \sum_{w = 1}^{\ell_0} \# \left\{I:|I| > \frac{w \pi}{\sqrt{\epsilon}} - 1\right\}.
\end{equation}

So far, we have been considering the system with a fixed length $L$.  In such system, the number of the intervals of zero potential (separated by one or more sites at which $V(x) = b$) is random.  Throughout the paper, the word "interval" means a set of consecutive sites where the potential function is equal to zero.  We now switch to a more convenient system, where the number $n$ of these intervals is fixed and, as a result, the size $L$ of the system becomes random.  In the infinite volume limit ($L, n \to \infty$) the two systems are equivalent.  This approach was already used, and is explained in detail in  \cite{BishopWehr12}.

Let us denote by $S_n(Y)$ the random variable counting the number of intervals whose length is larger than $Y$ in a system with the total number of intervals equal $n$. By the Glivenko-Cantelli theorem (see \cite{Durrett}, Theorem 2.4.7), we have
\begin{equation}																											\label{eq:NLineq2}
S_n\left(\frac{w \pi}{\sqrt{\epsilon}} - 1\right) = \# \left\{I:|I| > \frac{w \pi}{\sqrt{\epsilon}} - 1\right\} = P\left[|I| > \frac{w \pi}{\sqrt{\epsilon}} - 1\right]n + R(n),
\end{equation}
where the remainder $R(n) = o(n)$ as $n \to \infty$ uniformly in $\frac{w \pi}{\sqrt{\epsilon}}$. To finish the proof, however, we need more precise information about the behavior of the remainder $R$ at infinty. In fact, we will show
\begin{equation}																			\label{eq:Rnlimit}
\lim_{n\to\infty} \frac {R(n)}{n^{\gamma}} = 0 \quad a.s.
\end{equation}
for $\gamma \in (\frac 1 2 ,  1)$, uniformly in $w$ and $\epsilon$. Indeed, for a positive number $\eta$ let us estimate
$
P\left[\frac {R(n)}{n^\gamma} > \eta\right].
$
Equation (\ref{eq:NLineq2}) implies
\begin{equation}																							\label{eq:etaprob}
P\left[\frac {R(n)}{n^\gamma} > \eta\right] = P\left[\frac {S_n - \rho n}{n^\gamma} > \eta\right] = P\left[S_n > \rho n + \eta n^\gamma\right],
\end{equation}
where we use a simplified notation $S_n = S_n\left(\frac{w \pi}{\sqrt{\epsilon}} - 1\right)$ and $\rho = P\left[|I| > \frac{w \pi}{\sqrt{\epsilon}} - 1\right]$.
Applying the exponential Chebyshev inequality, we get
$$
P[S_n \geq \rho n + \eta n^{\gamma}] \leq \exp[-t(\rho n + \eta n^{\gamma})]M(t)^n,
$$
where 
$$
M(t) = \rho e^t + (1-\rho)
$$
is the moment generating function of the Bernoulli distribution with the parameter $\rho$.  Expanding the exponential to second order, we get
$$
\left[M(t)\right]^n \le e^{n\ln\left(1 + \rho (t + \frac e 2 t^2)\right)} \le e^{n\rho (t + \frac e 2 t^2)},
$$
whence for the probability (\ref{eq:etaprob}) we obtain
$$
P\left[\frac {R(n)}{n^\gamma} > \eta\right] \leq \exp\left[-t\eta n^{\gamma}+n\rho \frac e 2 t^2\right] \le \exp\left[-\frac {\eta^2 n^{2\gamma -1}}{2 \rho e}\right].
$$
Since the right-hand side of the above inequality is summable in $n$, by the First Borel-Cantelli lemma we get that for $n$ larger than a certain $n_0$ (depending on $\omega$)
$$
\frac{R(n)} {n^{\gamma}} \le \eta.
$$
Similarly, one shows that for large $n$ 
$$
\frac {R(n)}{n^{\gamma}} \ge -\eta,
$$
which proves (\ref{eq:Rnlimit}).

Combining (\ref{eq:NLineq1}) and (\ref{eq:NLineq2}) with (\ref{eq:Rnlimit}), we obtain $a.s.$ as $n \to \infty,$
\begin{equation}																													\label{eq:NLineq}	
N_L(\epsilon) \ge \sum_{w = 1}^{\ell_0} P\left[|I| > \frac{w \pi}{\sqrt{\epsilon}} - 1\right]n + o(n^{\gamma}) = \sum_{w = 1}^{\ell_0} P\left[|I| > \left\lfloor \frac{w \pi}{\sqrt{\epsilon}}\right\rfloor - 1\right]n + o(n^{\gamma}),
\end{equation}
where the last term estimates $R(n)$ multiplied by $\ell_0$ and, to accommodate the (at most logarithmic, see Proposition~\ref{prop:ell0length}) growth of $\ell_0$ we have to increase the value of $\gamma$ slightly.
Using that in the system with a fixed number of intervals the interval lengths are independent random variables with geometric distribution, it is easy to compute that
\begin{equation}																											\label{eq:prob1}
P\left[|I| > \left\lfloor \frac{w \pi}{\sqrt{\epsilon}}\right\rfloor - 1\right] = p^{ \left\lfloor \frac{w \pi}{\sqrt{\epsilon}}\right\rfloor -1}.
\end{equation}
Therefore,
\begin{align}
\label{eq:CFlowerbnd} N_L(\epsilon) &\ge \sum_{w = 1}^{\ell_0} p^{ \left\lfloor \frac{w \pi}{\sqrt{\epsilon}}\right\rfloor -1} p\,qL + o(L^\gamma) = \sum_{w = 1}^{\ell_0} p^{\frac{w \pi}{\sqrt{\epsilon}} -1} p^{\left\lfloor \frac{w \pi}{\sqrt{\epsilon}}\right\rfloor - \frac{w \pi}{\sqrt{\epsilon}}} p\,qL + o(L^\gamma)\\
\notag &\ge \sum_{w = 1}^{\ell_0} p^{\frac{w \pi}{\sqrt{\epsilon}}} p^{-2} p\,qL + o(L^\gamma)
= \frac {p^{\frac{\pi}{\sqrt{\epsilon}}}(1 - p^{\frac{\pi \ell_0}{\sqrt{\epsilon}}})}{1 - p^{\frac{\pi}{\sqrt{\epsilon}}}}p^{-1} qL + o(L^\gamma),
\end{align}
finishing the proof of (\ref{eq:LifTailslowerbnd}).
\end{proof}

\begin{proposition} 													\label{prop:ell0length}
Let $\ell_0$ be the random variable, measuring length of the longest interval of zero potential. Then
\begin{equation}																									\label{eq:elllimitinfty}
\lim_{n\to\infty} \ell_0 = \infty \quad a.s.
\end{equation}
and
\begin{equation}																									\label{eq:elllimit}
\lim_{n \to \infty} \frac {\ell_0}{n^\delta} = 0 \quad a.s.
\end{equation}
for any $\delta > 0$. 
\end{proposition}
\begin{proof} 
Both statements follow from the fact that the length of the largest interval grows logarithmically with the size of the system. Specifically, 
a simple calculation shows that for any $0 \le y \le \infty$,
\begin{equation}																										\label{eq:loggrowth}
\lim_{n \to \infty} P\left[\ell_0 > \frac{\log n}{\log \frac 1 p} - \frac{\log y}{\log \frac 1 p}\right] = 1 - e^{-y}.
\end{equation}
Both (\ref{eq:elllimitinfty}) and (\ref{eq:elllimit}) are straightforward consequences of (\ref{eq:loggrowth}). 
\end{proof}

\begin{proof}[Proof of (\ref{eq:LifTailslowerbnd}).]
As $L \to \infty$, $\ell_0 \to \infty$ almost surely, therefore
$$\lim_{L \to \infty} p^{\frac{\pi \ell_0}{\sqrt{\epsilon}}} = 0, \quad a.s.$$
With this observation, (\ref{eq:NLlowerbnd}) implies (\ref{eq:LifTailslowerbnd}) immediately.
\end{proof}

\section{Upper Bound on the Lifschitz Tail}
\label{sec:4}

To prove an upper bound on the integrated density of states $k(\epsilon)$, we are first going to prove that low lying excited states are mostly concentrated on the sites with zero potential, i.e. on the intervals.  We will then derive a lower bound on the length of intervals that can support an eigenstate with energy smaller than $\epsilon$.  The total number of such eigenstates will then be estimated using the geometric distribution of individual intervals, similarly to a calculation in section 2 (equations (\ref{eq:NLineq1})--(\ref{eq:CFlowerbnd})).


Consider an eigenstate $f$.  Its restriction to an interval of zero potential $I = (x_0, x_0 +\ell + 1)$ of length $\ell$ has the form
\begin{equation}															\label{eq:exstate}
f(x) = \frac{c}{\sqrt{\ell + 1}} \sin\left(\frac {\alpha \pi (x - x_0)}{\ell + 1} + t\right), \quad x \in I,
\end{equation}
with the boundary values that agree with (\ref{eq:exstate}):
\begin{equation}															\label{eq:exstatebnd}
f(x) = 
\begin{cases} 
\delta_L = \frac{c}{\sqrt{\ell + 1}} \sin\left(t\right), & x = x_0,\\
\delta_R = \frac{c}{\sqrt{\ell + 1}} \sin\left(\alpha \pi + t\right), & x = x_0 + \ell + 1.
\end{cases}
\end{equation}
The normalization constant $c$ is chosen in such a way that $\|f\|_{l_2(I)} = 1$. The phase shift $t$ takes values between $-\pi/2$ and $\pi/2$ and is determined by the left boundary value. Note that, since the values of $f$ on the boundary of the interval are not necessarily zero, $\alpha \ge 0$ does not have to be an integer and $f$ is a distortion of a Dirichlet eigenfunction of the Laplacian on $I$.  

The following properties of $f$ will be useful for the future estimates.
\begin{lemma}
Let $f$ be defined by (\ref{eq:exstate}). Then 
\begin{equation}										\label{eq:cge1}
c \ge 1.
\end{equation}
Furthermore, let $\delta = \max \{|\delta_L|, |\delta_R|\}$. Then 
\begin{equation}										\label{eq:alphaineq}
\alpha \ge [\alpha] + 1 - \delta \sqrt{\ell + 1} \quad \textrm{ if } \quad \{\alpha\} \pi + t > \pi/2, 
\end{equation}
where $[\alpha]$ stands for the integer part of $\alpha$ (the largest integer not exceeding $\alpha$) and $\{\alpha\} = \alpha - [\alpha]$ is its fractional part.
\end{lemma}
\begin{proof} The first statement follows immediately from the observation that: 
\begin{equation}
1 = \|f\|_{l_2(I)}^2 = \frac{c^2}{\ell} \sum_{x = x_0 + 1}^{x_0 + \ell} \sin^2\left(\frac {\alpha \pi (x - x_0)}{\ell + 1} + t\right) \le \frac{c^2}{\ell}\, \ell = c^2.
\end{equation}
To prove inequality (\ref{eq:alphaineq}), let us solve both equations (\ref{eq:exstatebnd}) simultaneously with respect to $t$. Since $t$ is between $-\pi/2$ and $\pi/2$, the first equation becomes 
\begin{equation}												\label{eq:tL}
t = \arcsin \left(\frac {\delta_L \sqrt{\ell + 1}}{c}\right),
\end{equation}
and since $ \{\alpha\} \pi + t > \pi/2$, from the second equation we obtain the following expression for $t$:
\begin{align}
\label{eq:tR1}&t = (-1)^{[\alpha] +1}\arcsin \left(\frac {\delta_R \sqrt{\ell + 1}}{c}\right) - \{\alpha\} \pi + \pi.
\end{align}
Next, set the right-hand sides of (\ref{eq:tL}) and (\ref{eq:tR1}) equal to each other to get
\begin{align*}
\alpha &= [\alpha] + 1 - \frac 1 \pi \left(\arcsin \left(\frac {\delta_L \sqrt{\ell + 1}}{c}\right) + (-1)^{[\alpha]}\arcsin \left(\frac {\delta_R \sqrt{\ell + 1}}{c}\right)\right)\\
& \ge [\alpha] + 1 - \frac 1 \pi \left(\arcsin \left(\frac {|\delta_L| \sqrt{\ell + 1}}{c}\right) + \arcsin \left(\frac {|\delta_R| \sqrt{\ell + 1}}{c}\right)\right)\\
& \ge [\alpha] + 1 - \frac 2 \pi \arcsin \left(\frac {\delta \sqrt{\ell + 1}}{c}\right).
\end{align*}
Since for positive $x$,
$$\arcsin x \le \frac \pi 2 x,$$
we have
\begin{equation}
\alpha \ge [\alpha] + 1 - \frac {\delta \sqrt{\ell + 1}}{c} \ge [\alpha] + 1 - \delta \sqrt{\ell + 1},
\end{equation}
where in the last inequality we used (\ref{eq:cge1}). This finishes the proof.


\end{proof}

The purpose of lemma is to provide lower bounds on $\alpha$ to be used in lower bounds on the energies of states.  The goal is to find optimal lower bounds on energy for a sine wave with approximate frequency $w$.  This approximate frequency is determined by the branch of $\arcsin x$.  In the case where $\{\alpha\}\pi + t > \pi/2$, the function in (\ref{eq:exstate}) is a transformation of the sine wave with integer frequency $w = [\alpha]+1$ that is stretched and shifted.  This case describes a state with energy strictly less than, but close to, the energy of the sine wave with frequency $w$.  In the case where $\{\alpha\}\pi + t \leq \pi/2$, the function in (\ref{eq:exstate}) is a transformation of the sine wave with integer frequency $w = [\alpha]$ that is compressed and shifted.  This case describes a state with energy greater than the energy of the state with frequency $w$.  This energy is not less than the energy of state with integer frequency $w$ and not a lower bound.  Thus, it must be the case that $\{\alpha\}\pi + t > \pi/2$, which will be the case considered in Lemma 2.




\begin{lemma}
Let $I = (x_0 , x_0 +\ell + 1)$ be an interval of length $\ell$ and let $f$ be an eigenstate supported on $I$ of the form (\ref{eq:exstate}).
Then for $\ell + 1 \ge \pi^2/b$,
\begin{equation}																	\label{eq:islandlengthineq}
\ell + 1 \ge \frac{w\pi}{\sqrt{\epsilon}}- \frac {\pi^2}{b} + (w + \frac 1 w) O(\sqrt{\epsilon}) + O(\epsilon), \quad \textrm{ as } \epsilon \to 0,
\end{equation}
where $w = [\alpha] + 1$ and $\epsilon$ is the energy of $f$.\\
\end{lemma}
\begin{proof}
First notice that the energy of $f$ satisfies the inequality
\begin{equation}
\epsilon \ge 4 \sin^2 \left(\frac {\alpha \pi}{2(\ell + 1)}\right) + b\,(\delta_L^2 + \delta_R^2),
\end{equation}
where the first term represents the kinetic (and total) energy of the sine wave over $I$ and the second term is the potential energy of the endpoints. 
Using (\ref{eq:alphaineq}) we obtain
\begin{equation}																		\label{eq:epsilonineq1}
\epsilon \ge 4 \sin^2 \left((w - \delta \sqrt{\ell + 1}) \frac {\pi}{2(\ell + 1)}\right) + b\delta^2,
\end{equation}
where $\delta = \max \{|\delta_L|, |\delta_R|\}$, as in (\ref{eq:alphaineq}). In order to proceed to the next step notice that if $\sin^2 x \le \epsilon$, then 
\begin{equation}																															\label{eq:sineineq}
\sin^2(x) \ge x^2 (1 + O(\epsilon)), \quad \text { as } \epsilon \to 0.
\end{equation}
Indeed, since
$$\sin^2 x \ge x^2 - \frac 1 3 \, x^4 = x^2\left(1 - \frac 1 3 \, x^2\right), \quad \textrm{ for }  0 \le x \le \pi, $$
and
$$ x^2 \le \frac {\pi^2}{4} \sin^2 x \le \frac {\pi^2}{4}\, \epsilon, \quad \textrm{ for }  0 \le x \le \pi/2,$$
we get 
$$\sin^2 x \ge x^2\left(1 - \frac {\pi^2}{12}\, \epsilon\right), \quad \textrm{ for }  0 \le x \le \pi/2,$$
which implies (\ref{eq:sineineq}). 

After applying inequality (\ref{eq:sineineq}) to (\ref{eq:epsilonineq1}) we have
\begin{equation}																														\label{eq:epsilonineq2}
\epsilon \ge \left(1 - \frac{\delta \sqrt{\ell + 1}} {w}\right)^2 \frac {w^2\pi^2}{(\ell + 1)^2}(1 + O(\epsilon)) + b\delta^2, \quad \textrm{ as } \epsilon \to 0.
\end{equation}
In other words, for an interval $I$ to support a distorted sine wave of frequency $\alpha$ (between $w-1$ and $w$), with the absolute value of the largest boundary condition $\delta$ and energy $\epsilon$, its length $\ell$ must satisfy (\ref{eq:epsilonineq2}). 

Let us find the boundary conditions $\delta$ that minimize the left-hand side of (\ref{eq:epsilonineq2}), which represents a quadratic expression with the positive leading coefficient with respect to $\delta$. The global minimum of this function is attained at
\begin{equation}																					
\delta = \frac {w\pi^2(1 + O(\epsilon))}{(\ell + 1)^{3/2}\left(\frac {\pi^2}{\ell + 1}(1 + O(\epsilon)) + b\right)},
\end{equation}
which, substituted into the right-hand side of (\ref{eq:epsilonineq2}), produces the inequality
\begin{equation}																								\label{eq:epsilonineq3}
\epsilon \ge \left(1 - \frac {\pi^2(1 + O(\epsilon))}{b(\ell + 1)}\right)^2 \frac {w^2\pi^2}{(\ell + 1)^2}(1 + O(\epsilon)), \quad \textrm{ as } \epsilon \to 0,
\end{equation}
for $\ell$ sufficiently large ($\ell + 1 \ge \pi^2/b$ is enough). Also, in the transition from (\ref{eq:epsilonineq2}) to (\ref{eq:epsilonineq3}) we dropped the positive term of $b\delta^2$. Again, we get that for an interval to support a state of approximate frequency $w$ with energy $\epsilon$, its length $\ell$ must satisfy (\ref{eq:epsilonineq3}). 

Let us find an explicit bound on $\ell$. Using that 
$$\sqrt{1 + O(\epsilon)} = 1 + O(\epsilon) \quad \textrm{ and } (1 + O(\epsilon))^2 = 1 + O(\epsilon), \quad \textrm{ as } \epsilon \to 0,$$
we obtain a quadratic equation on $\ell + 1$:
\begin{equation}
\sqrt{\epsilon}(\ell + 1)^2 - w\pi(1 + O(\epsilon))(\ell + 1) + \frac {w\pi^3}{b} (1 + O(\epsilon)) \ge 0, \quad \textrm{ as } \epsilon \to 0.
\end{equation}
Since the smaller root of the corresponding quadratic equation is negative, the above inequality only holds if $\ell + 1$ is larger than the bigger root,
\begin{equation}
\ell + 1 \ge \frac{w\pi(1 + O(\epsilon)) + \sqrt{w^2\pi^2(1 + O(\epsilon)) - \frac {4w\pi^3}{b} \sqrt{\epsilon}(1 + O(\epsilon)) }}{2\sqrt{\epsilon}}, \quad \textrm{ as } \epsilon \to 0.
\end{equation}
Using the Taylor series
$$\sqrt{x_0 - x} = \sqrt{x_0} - \frac {x}{2\sqrt{x_0}} + O\left(\frac {x^2}{x_0^{3/2}}\right),$$
we get the claimed condition on the interval length
\begin{equation}
\ell + 1 \ge \frac{w\pi}{\sqrt{\epsilon}}- \frac {\pi^2}{b} + (w + \frac 1 w) O(\sqrt{\epsilon}) + O(\epsilon), \quad \textrm{ as } \epsilon \to 0.
\end{equation}

\end{proof}

\begin{lemma} With probability one
\begin{equation}																							\label{eq:NLupperbnd}
N_L(\epsilon) \le \frac{q\left(1 - p^{\ell_0\frac{\pi}{\sqrt{\epsilon}} + \ell_0 O(\sqrt{\epsilon})}\right)p^{\frac{\pi}{\sqrt{\epsilon}}- \frac {\pi^2}{b}}}{p^2\left(1-p^{\frac{\pi}{\sqrt{\epsilon}} + O(\sqrt{\epsilon})}\right)}L  + o(L^\gamma), \quad \textrm{ as } L \to \infty \textrm{ and } \epsilon\to 0,
\end{equation}
where $1/2 < \gamma < 1$ and $\ell_0 (= \ell_0(\omega))$ is the length of the longest interval of zero potential.
\end{lemma}
\begin{proof}
Each state is an eigenfunction of the form (\ref{eq:exstate}) when restricted to an interval of zero potential.  The number of possible frequencies that such a state could take on a given interval is counted by the variable $w$. For the state to have energy less than $\epsilon$, the interval must necessarily satisfy the lower bound derived in Lemma 2.   The dimension of the space of states with energy less than $\epsilon$ is bounded by counting the maximum number of possible frequencies such a state could take on a given interval of zero potential and summing these maxima over each interval of zero potential that could support at least one such state.

Similar to equation (\ref{eq:NLineq}), for a fixed $ \gamma \in (\frac 1 2 , 1)$, we have as $n \to \infty$, uniformly in $\omega$ and $\epsilon$:
\begin{align*}
N_L(\epsilon) &\le \sum_{w = 1}^{\ell_0} \# \left\{I: |I| > \frac{w\pi}{\sqrt{\epsilon}}- \frac {\pi^2}{b} -1 + (w + \frac 1 w) O(\sqrt{\epsilon}) + O(\epsilon)\right\}\\
&= \sum_{w = 1}^{\ell_0} P\left[|I| > \frac{w\pi}{\sqrt{\epsilon}}- \frac {\pi^2}{b} -1 + (w + \frac 1 w) O(\sqrt{\epsilon}) + O(\epsilon)\right]n +  o(n^{\gamma})\\
&= \sum_{w = 1}^{\ell_0} P\left[|I| > \left\lfloor \frac{w\pi}{\sqrt{\epsilon}}- \frac {\pi^2}{b} + (w + \frac 1 w) O(\sqrt{\epsilon}) + O(\epsilon)\right\rfloor -1\right]n +  o(n^{\gamma}),
\end{align*}
Further, similarly to (\ref{eq:prob1}) and (\ref{eq:CFlowerbnd}) we have
\begin{equation}																											
P\left[|I| > \left\lfloor \frac{w\pi}{\sqrt{\epsilon}}- \frac {\pi^2}{b} + (w + \frac 1 w) O(\sqrt{\epsilon}) + O(\epsilon)\right\rfloor - 1\right] = p^{ \left\lfloor \frac{w\pi}{\sqrt{\epsilon}}- \frac {\pi^2}{b} + (w + \frac 1 w) O(\sqrt{\epsilon}) + O(\epsilon)\right\rfloor -1}
\end{equation}
and
\begin{align}
\notag N_L(\epsilon) &\le \sum_{w = 1}^{\ell_0} p^{\frac{w\pi}{\sqrt{\epsilon}}- \frac {\pi^2}{b} + (w + \frac 1 w) O(\sqrt{\epsilon}) + O(\epsilon)} p^{-2} p\,qL +  o(L^{\gamma})\\
\label{eq:countfncineq}&\le p^{- \frac {\pi^2}{b}}  p^{-1.5}qL \sum_{w = 1}^{\ell_0} p^{w\left(\frac{\pi}{\sqrt{\epsilon}} + O(\sqrt{\epsilon})\right)}+  o(L^{\gamma}) \\
\notag &= p^{- \frac {\pi^2}{b}}  p^{-2}qL \frac{p^{\frac{\pi}{\sqrt{\epsilon}}}\left(1 - p^{\ell_0\left(\frac{\pi}{\sqrt{\epsilon}} + O(\sqrt{\epsilon})\right)}\right)}{1-p^{\frac{\pi}{\sqrt{\epsilon}} + O(\sqrt{\epsilon})}} +  o(L^{\gamma}), \quad \textrm{ as } L\to\infty \textrm{ and } \epsilon\to 0,
\end{align}
completing the proof of (\ref{eq:LifTailsupperbnd}).
\end{proof}

\begin{proof}[Proof of (\ref{eq:LifTailsupperbnd}).]
It is obvious from (\ref{eq:countfncineq}) that for a fixed small enough $\epsilon$ 
\begin{equation}
\limsup_{L\to\infty} \frac{N_L(\epsilon)}{L} \le \frac{q p^{\frac{\pi}{\sqrt{\epsilon}}- \frac {\pi^2}{b}}}{p^2\left(1-p^{\frac{\pi}{\sqrt{\epsilon}} + C(\epsilon)}\right)},
\end{equation}
where the constant $C(\epsilon)$ behaves as $O(\sqrt{\epsilon})$ as $\epsilon \to 0$.
\end{proof}

\def\cprime{$'$}

\section{Acknowledgements}
The authors would like to thank R. Sims, L. Friedlander, and K.McLaughlin, A. Fedorenk for useful discussions.  M. Bishop and J. Wehr were partly supported by NSF grant DMS 0623941.

\end{document}